\documentclass[conference]{IEEEtran}
\IEEEoverridecommandlockouts
\usepackage{cite}
\usepackage{amsmath,amssymb,amsfonts}
\usepackage{float}
\usepackage{graphicx}
\usepackage{textcomp}
\usepackage{xcolor}
\usepackage{amsthm}
\usepackage{enumitem}
\newtheorem{theorem}{Theorem}
\usepackage{todonotes}
\def\BibTeX{{\rm B\kern-.05em{\sc i\kern-.025em b}\kern-.08em
    T\kern-.1667em\lower.7ex\hbox{E}\kern-.125emX}}
\usepackage{algorithm2e}
\SetKwComment{Comment}{/* }{ */}
\usepackage{amsmath}

\DeclareMathOperator*{\argmin}{arg\,min}

\RestyleAlgo{ruled}
    
\newtheorem{prop}{Proposition}

\newcommand{\ym}[1]{\textcolor{black}{#1}}

\newcommand{\sm}[1]{\textcolor{black}{#1}}

\begin{document}

\title{Alternative paths computation for congestion mitigation in segment-routing networks
}

\author{\IEEEauthorblockN{Sébastien Martin ~~~~~~~ Youcef Magnouche ~~~~~~~ Paolo Medagliani ~~~~~~~ Jérémie Leguay } 
\IEEEauthorblockA{\textit{Huawei Technologies Ltd., Paris Research Center, France} \ \\
firstname.lastname@huawei.com}
}

\maketitle

\begin{abstract}
In backbone networks, it is fundamental to quickly protect traffic against any unexpected event, such as failures or congestions, which may impact Quality of Service (QoS). Standard solutions based on Segment Routing (SR), such as Topology-Independent Loop-Free Alternate (TI-LFA), are used in practice to handle failures, but no distributed solutions exist for distributed and tactical congestion mitigation. A promising approach leveraging SR has been recently proposed to quickly steer traffic away from congested links over alternative paths. As the pre-computation of alternative paths plays a paramount role to efficiently mitigating congestions, we investigate the associated path computation problem aiming at maximizing the amount of traffic that can be rerouted as well as the resilience against any 1-link failure. In particular, we focus on two variants of this problem. First, we maximize the residual flow after all possible failures. We show that the problem is NP-Hard,  and we solve it via a Benders decomposition algorithm. Then, to provide a practical and scalable solution, \ym{we solve a relaxed variant problem, that maximizes, instead of flow}, the number of surviving alternative paths after all possible failures. We provide a polynomial algorithm. Through numerical experiments, we compare the two variants and show that they allow to increase the amount of rerouted traffic and the resiliency of the network after any 1-link failure.
\end{abstract}


\section{Introduction}
Telecommunication operators constantly face a dilemma when they need to operate their networks. On one side, they need to keep the utilization as high as possible, to maximize the return of investments. On the other side, they need to keep the utilization as low as possible to provide high Quality of Service (QoS) to users. 
Efficient network planning helps in finding the right trade-off by playing on capacity sizing and routing strategies. The link capacities can be designed to support a given estimated amount of traffic in the network, avoiding saturation. However, as capacities cannot be infinite or easily adjusted when traffic conditions evolve, choosing a suitable routing strategy is also crucial.

In the largest majority of backbone networks, Best Effort (BE) traffic follows the shortest paths from Interior Gateway Protocols (IGP), where link metrics are typically designed to be inversely proportional
to the capacity, to capture the willingness of the operator
to attract traffic over large-capacity links. On top of that,  flows with specific requirements, for instance in terms of QoS, can be "engineered" to follow specific paths, different from default ones. Segment Routing (SR)~\cite{ventre2020segment} provides a powerful solution to steer engineered traffic by defining intermediary segments or way-points that the packets must follow.

Despite that optimal network planning can be achieved for a given set of expected traffic scenarios, the experienced amount of traffic can vary during normal network operations, for instance, due to unexpected bursts of traffic (i.e., large data transfers, video streaming, ...) or link failures, which can frequently happen in practice. 
Typically, in order to avoid costly over-provisioning, the network capacity is planned to support the failure of one link based on plausible traffic scenarios.
However, as traffic evolves over time and multiple failures can happen, congestion can still be introduced on some links, introducing packet loss and degrading QoS. 

An efficient way to mitigate congestion is to leverage a distributed Congestion Mitigation (CM) mechanism, such as the one recently introduced in~\cite{10302752}, which locally steers part of the traffic away from a congested link over alternative paths. When a router detects congestion on an adjacent link, a portion of the traffic is automatically offloaded and load balanced over a set of alternative paths using UCMP (Unequal Cost Multi Paths). The goal is to select lightly loaded paths to reroute a maximum of traffic, mitigate the congestion, and avoid creating new congestion elsewhere in the network. As congestions need to be quickly handled,  alternative paths are actually pre-computed without knowing loads at congestion time and selected, on-the-fly, with particular weights (or split ratios) based on actual traffic conditions to load balance traffic.
This mechanism is efficient as it quickly reacts and reroute incoming traffic. However, it requires a careful path design, in order to avoid introducing further congestion in the network or using alternative paths with failed links. A simple approach, such as K-Shortest Path (KSP)~\cite{yen1971finding}, is not effective as it does not take into account diversity among paths to load balance traffic. The Neighbor Deviation Algorithm (NDA), presented in~\cite{10302752}, based on an extension of the Dijkstra shortest path, optimizes the local diversity around the node detecting the congestion but does not consider remote diversity. Also, none of these solutions consider resilience against failures.

In this paper, we extend the work of~\cite{10302752} by considering the maximum flow that can be potentially rerouted over alternative paths after any 1-link failure. Indeed, the main scenario considered in practical deployments is that congestions are generally introduced after the failure of a link. We show that the associated pre-computation problem is NP-Hard, and we solve it via a Benders decomposition algorithm. Then, to provide a practical and scalable solution, \ym{we solve a relaxed variant problem, that maximize, instead of flow}, the number of surviving alternative paths after all possible failures. We provide a polynomial algorithm. Through numerical experiments,
we compare the two variants and show they allow to increase the amount of rerouted traffic and the resiliency of the network
after any 1-link failure.


The structure of the paper is the following. Sec.~\ref{sec:RelatedWork} presents the relevant state of the art and Sec.~\ref{sec:CongestionMitigaiton} describes the congestion mitigation mechanism our work is extending. In Sec.~\ref{sec:AlternativePathComputaion}, we give a formal description of the alternative paths' computation problem and provide the associated mathematical model. Then, we derive a Benders Decomposition algorithm for the problem. In Sec.~\ref{sec:Problem_relaxation}, we propose a slightly different variant of the problem that is polynomial to solve. Finally, in Sec.~\ref{sec:Experimentation} we show computational results. Sec.~\ref{sec:Conclusion} concludes this paper.
 
\section{Related Work}\label{sec:RelatedWork}

Several mechanisms have been proposed for tactical uses of SR. TI-LFA has been proposed to handle single failures (nodes, link or SRLG~\cite{ietf-rtgwg-segment-routing-ti-lfa-12, 10108280}). For multiple failures, \cite{foerster2018ti} proposed to compute several backup paths to consider different failure scenarios. 
SR extensions have also been proposed for micro-loop avoidance~\cite{bashandy-rtgwg-segment-routing-uloop-15, hegde-rtgwg-microloop-avoidance-using-spring-03}.



{
Although, max-flow \cite{goldberg1988new} and min-cost flow \cite{10.1145/76359.76368} problems can be solved in polynomial time, bounding the number of paths that the flow can take makes them harder to solve. In \cite{6195830}, authors investigate the problem of decomposing a flow into a minimum number of paths. They show that the problem is NP-hard and analyze the performance of a greedy approach. In \cite{koch2008maximum}, they investigate a similar problem, called the maximum k-splittable flows problem, that consists in computing exactly $k$ paths that maximize the flow between given source and destination nodes. They show that the problem is NP-hard for undirected graphs with $k=2$. In \cite{TRUFFOT2008629}, the authors propose a compact model for the maximum k-splittable flows problem and derive an extended formulation thanks to the Dantzig-Wolfe decomposition, for which they developed a Branch-and-Price algorithm. \\
However, considering $k\geq 2$ paths does not guarantee that the flow can be rerouted after any 1-link failure, as the paths can use the same link. This requires considering additional resiliency constraints on top of the maximum $k$-splittable flows problem. In~\cite{bendali2007k}, the authors investigate the $k$-edge connected subgraph problem that consists in computing $k + 1$ totally disjoint paths between each pair of nodes, which is resilient to any $k$ link failures. } \\
\sm{Resilience against failures can be seen as an interdiction or a blocker version of the max flow problem. The interdiction maximum flow problem consists in finding a set of links under a given budget, such that removing them reduces the remaining maximum flow. The blocker maximum flow problem consists in finding the minimum number of link removal that reduces the remaining maximum flow below a given threshold. 
In~\cite{MFI}, the authors propose an integer linear formulation to solve the interdiction version of the max flow problem, and in~\cite{MFB} they consider the blocker version, where they propose a branch-and-benders cut to tackle it. }
{In contrast to existing works, we focus on the worst-case scenario, i.e.,  the most impacting 1-link failure in terms of resiliency and flow to be rerouted over alternative paths.}

\section{Congestion Mitigation} \label{sec:CongestionMitigaiton}
As presented in~\cite{10302752}, the CM mechanism applies to backbone networks where Provider (P) nodes route traffic coming from Provider Edge (PE) nodes at the boundaries of the network.  At network startup, each router pre-computes up to $K$ alternative paths towards each PE node and for each potentially adjacent congested link. As SR is used for routing, the alternative paths are translated into Segment ID (SID) lists, stored in the routing table of the P nodes and activated on the fly when needed.

As soon as the load of a link exceeds a given threshold $Th_{\rm H}$, i.e. 70\%, the associated P node locally offloads traffic over the pre-computed alternative paths.
For each incoming packet to be rerouted, the P node stitches into the SR Header (SRH), the SID list of the chosen alternative path and selects the proper outgoing interface, forwarding the packet on it.  

The traffic is load-balanced by assigning Unequal Cost Multiple Path (UCMP) weights to each alternative path. As it is not possible to determine the weights to be used in the case of congestion, they are decided when congestion mitigation is activated according to actual traffic conditions. A second threshold, referred to as target threshold ($Th_{\rm T}$), indicates the target link load (e.g., 50\%) to reach after mitigation on the congested link. It is used 
to determine the amount of traffic to offload. More details on the computation of weights can be found in~\cite{10302752}.


As traffic cannot be known in advance, 
pre-computed alternative paths must have the following properties:
\sm{\begin{itemize}
    \item \emph{maximal available bandwidth}: 
    the alternative paths must maximize the flow that can be rerouted through them for each possible 1-link failure. 
    \item \emph{routing cost}: the alternative paths must minimize the routing cost.
\end{itemize}
}
\sm{In the next section, we will describe how to achieve these criteria. For the description of the algorithms, we assume that the congested link is removed from the network, and we directly consider a source, i.e., the node detecting the congestion, and a destination, i.e., a PE router. In~Sec.~\ref{sec:Experimentation}, we will detail how the source and the destination are selected.}

\section{Alternative paths' computation}\label{sec:AlternativePathComputaion}
Let $G=(V,A)$ be a directed graph where $V$ is the set of vertices and $A$ is the set of links. For each link $a\in A$, the capacity is denoted $b_a$ and the cost $c_a$. 
Let $c(p)$ be the cost of a path $p$, i.e., $c(p)=\sum_{a\in p} c_a$.
We define $mincap(p)=\min_{a\in p} b_a$. Let $G\setminus A'$, $A'\subseteq A$, be the subgraph of $G$ where links of $A'$ are removed.  \\
The alternative paths' computation problem (APCP) consists in finding a set $P$ of $k$ paths between $s$ and $t$, where $s,t\in V$ are the source and destination nodes respectively, such that
\begin{itemize}[leftmargin=5.6em]
    \item[\textbf{Objective 1}:] the worst maximum flow over every possible 1-link failure is maximized, i.e., $\max \limits_{a\in A} maxflow(G\setminus \{a\})$,
     \item[\textbf{Objective 2}:] the total cost is minimized, i.e., $\min\sum \limits_{p\in P}c(p)$,
\end{itemize}
where these two objectives are considered in lexicographic order, i.e., objective 1) $>>$ objective 2). To manage the lexicographic order between the two objectives, we use a bigM value $M_1=k\sum_{a\in A}c_a$. This ensures that each unity of flow costs more than any path in the network. 
\begin{figure}[!t]
\centering
\includegraphics[page=1,scale=.20]{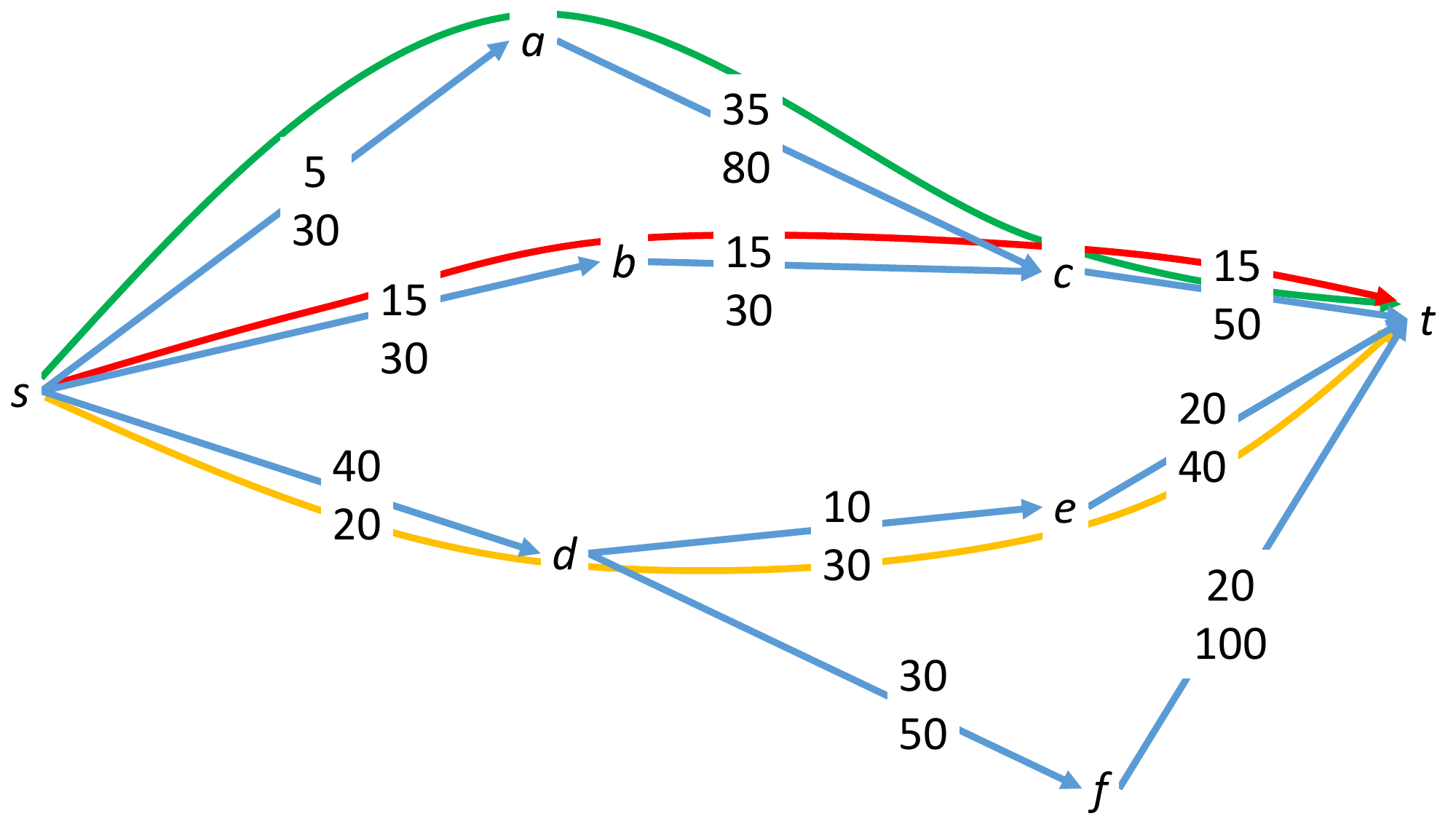}
\caption{Example of solution for APCP with $k=3$}
\label{APCP:fig}
\vspace{-5mm}
\end{figure}

Figure \ref{APCP:fig} shows a solution of the APCP with $k=3$. On each link, the top value is the cost and the bottom one is the capacity. The value of the objective 1) is a flow of 20, which corresponds to the maximum flow when the link $(c,t)$ fails. The value of the objective 2) is 170. If $k=2$ then the optimal solution is given by the red path and the yellow path, which are disjoint in this example. 

\begin{theorem}
The alternative paths' computation problem, already by only considering objective 1), is NP-hard, even for $k=3$. 
\end{theorem}
\begin{proof}
    In \cite{koch2008maximum}, the authors show that the maximum $k$-splittable flows problem (MkSFP) is NP-hard. The MkSFP consists in finding the maximum $s-t$ flow, split over at most $k$ paths. We can reformulate the problem by finding $k$ paths, such that the graph induced by these paths has a maximum flow. In the APCP, we consider the notion of failures in comparison to the MkSFP. By adding a direct link $(s,t)$ with a capacity equal to the max flow in the original graph, as the capacity of $(s,t)$ is sufficiently large, the optimal solution of the APCP in this new graph, where $k+1$ paths are considered, contains $(s,t)$. The remaining $k$ paths will be the optimal solution of the MkSFP, which ends the proof.
\end{proof}

Note that APCP can be solved in polynomial time if $k=1$, 
as it is equivalent to the shortest path problem since only objective 2) can be considered. 
For $k=2$ the complexity is an open question.

\subsection{Compact formulation}
In this section, we propose a mathematical model to solve the APCP. We consider the following decision variables
\begin{itemize}
    \item $z$: the worst maximum flow, among every 1-link failure scenario, that can be sent on alternative paths. 
    \item $x_a^i$: equals to 1 if the $i$-th path uses link $a$ and 0 otherwise.
    \item $f_a^{\bar{a}}$: equal to the quantity of flow sent on link $a$ when link $\bar{a}$ fails.
\end{itemize}

The APCP is equivalent to the following integer linear program (P):
\vspace{-2mm}
{\small
\begin{align}
& \max M_1 z -\sum_{i\in K}\sum_{a\in A} c_a x_a^i \label{variant1:Objectif}\\
 & z \leq \sum_{a\in \delta^+(s)} f_a^{\bar{a}} & & \forall \bar{a} \in A, \label{variant1:worstMaxFlow}\\
& f_a^{\bar{a}} \leq b_a y_a & & \forall a \in A, 
 \label{variant1:LinkCapacity}\\ 
& y_a \leq \sum_{i\in K} x^i_a& &  \forall a\in A,\label{variant1:linkOpening}\\
& \sum_{a \in \delta^+(v)} \hspace{-3mm} x_{a}^i - \hspace{-3mm} \sum_{a \in \delta^-(v)}\hspace{-3mm} x_{a}^i
		= \begin{cases} 
			1 &\text{ if } v = s, \\
			-1  &\text{ if } v = t, \\
			0 &\text{ otherwise;}
		\end{cases}
		\hspace{-6mm}& &~~~ \forall v \in V, i \in K, \label{variant1:pathComputation}\\	
& \sum_{a \in \delta^+(v)} x_{a}^i 
		\leq  \begin{cases}  
			0  &\text{ if } v = t, \\
			1 &\text{ otherwise;}
		\end{cases}
		\hspace{-6mm}& & \forall v \in V, i \in K, \label{variant1:degreeNodes1}\\	
& \sum_{a \in \delta^-(v)} x_{a}^i 
		\leq  \begin{cases}  
			0  &\text{ if } v = s, \\
			1 &\text{ otherwise;}
		\end{cases}
		\hspace{-6mm}& & \forall v \in V, i \in K, \label{variant1:degreeNodes2}\\	
& x^i(C)\leq |C| - 1 & & \forall i\in K, \forall \text{ sub-tour }C,\label{variant1:sub-tour-elimination}\\
& \sum_{a \in \delta^+(v)}  f_{a}^{\bar{a}} -\sum_{a \in \delta^-(v)} f_{a}^{\bar{a}}
		=  	0 		& & \forall v \in V\setminus \{s,t\}, \bar{a} \in A,\label{variant1:flowConservation} \\	
& f_{\bar{a}}^{\bar{a}}=0 & & \forall \bar{a}\in A.\label{variant1:flowFailure}
\end{align}
}
Constraints \eqref{variant1:worstMaxFlow} allow getting the worst maximum flow over all possible link failures. Constraints \eqref{variant1:LinkCapacity} bound the flow variables to the link capacity. Constraints \eqref{variant1:linkOpening} ensure that any link used by the flow belongs to an alternative path. Constraints \eqref{variant1:pathComputation}-\eqref{variant1:degreeNodes2} compute the paths between $s$ and $t$. Constraints \eqref{variant1:sub-tour-elimination} represent the sub-tour elimination inequalities. Constraints \eqref{variant1:flowConservation} 
 represent the flow conservation constraints. Constraints \eqref{variant1:flowFailure} guarantee that no flow crosses the failed link.
Sub-tour elimination constraints are dynamically generated. During the exploration of the Branch-and-Bound tree, when an integer solution is found, we iterate over each link $(u,v)\in A$ and \sm{check with a Depth-First-Search (DFS) algorithm in the sub-graph defined by the solution if a loop appears. If a loop is detected then we add the associated sub-tour $C$. } 
\subsection{Benders model}
The model can be decomposed using Benders decomposition methods. It consists in keeping the binary variables in the master problem and moving the continuous variables and the associated constraints in the sub-problem. The master problem, then, will consist of computing the $k$ alternative paths for each failure, while the sub-problem consists in evaluating the max-flow over the paths given from the master problem. 
\subsubsection{Sub-problem}
Given an input solution $y^*$, for $\bar a\in A$, the sub-problem can be formulated as follows:
\vspace{-2mm}
{\small
\begin{align}
& \max \sum_{a\in \delta^+(s)} f_a^{\bar{a}} \label{bendersMasterSub:Objectif} \\ 
 & f_a^{\bar{a}} \leq b_a y^*_a & & \forall a \in A, 
 \label{bendersMasterSub:LinkCapacity}\\ 
& \sum_{a \in \delta^+(v)}  f_{a}^{\bar{a}} -\sum_{a \in \delta^-(v)} f_{a}^{\bar{a}}
		=  	0 		& & \forall v \in V\setminus \{s,t\},\label{bendersMasterSub:flowConservation} \\	
& f_{\bar{a}}^{\bar{a}}=0. & & \label{bendersMasterSub:flowFailure} 
\end{align}
}
Let $\alpha$ be the dual variable vector associated with \eqref{bendersMasterSub:LinkCapacity}. Then, the dual objective function of the sub-problem is:
{\small
$$\sum \limits _{a\in A} \alpha_{a }b_a y^*_a $$    
}
\subsubsection{Master problem}
The master problem is formulated as follows:
{\small
\begin{align}
& \max M_1  z - \sum_{i\in K}\sum_{a\in A} c_a x_a^i   \label{bendersMaster:Objectif}\\ 
& y_a \leq \sum_{i\in K} x^i_a& &  \forall a\in A,\label{bendersMaster:linkOpening}\\
& \sum_{a \in \delta^+(v)} \hspace{-3mm} x_{a}^i - \hspace{-3mm} \sum_{a \in \delta^-(v)}\hspace{-3mm} x_{a}^i
		= \begin{cases} 
			1 &\text{ if } v = s, \\
			-1  &\text{ if } v = t, \\
			0 &\text{ otherwise;}
		\end{cases}
		\hspace{-6mm}& & ~~~\forall v \in V, i \in K, \label{bendersMaster:pathComputation}\\
& \sum_{a \in \delta^+(v)} x_{a}^i 
		\leq  \begin{cases}  
			0  &\text{ if } v = t, \\
			1 &\text{ otherwise;}
		\end{cases}
		\hspace{-6mm}& & \forall v \in V, i \in K, \label{bendersMaster:degreeNodes1}\\	
& \sum_{a \in \delta^-(v)} x_{a}^i 
		\leq  \begin{cases}  
			0  &\text{ if } v = s, \\
			1 &\text{ otherwise;}
		\end{cases}
		\hspace{-6mm}& & \forall v \in V, i \in K, \label{bendersMaster:degreeNodes2}\\	
& x^i(C)\leq |C| - 1 & & \forall i\in K, \forall \text{ sub-tour }C,\label{bendersMaster:sub-tour-elimination}\\
 & z \leq \sum \limits _{a\in A}     \alpha^i_{a }  b_a y_a & & \forall \bar{a} \in A, i\in \Gamma_{\bar a}. \label{bendersMaster:benders}
\end{align} 
}
where $\Gamma_{\bar a}$ represents the set of dual solutions of the sub-problem associated with $\bar a$. \\
Consider a solution $x^*$ of the master problem. Let $\bar a^*\in A$ be the arc associated with the minimum objective value among all optimal solutions of the sub-problems for all $\bar a\in A$, i.e., $\bar a^* = \argmin\limits _{\bar a\in A}\{ \sum_{a\in \delta^+(s)} f_a^{\bar{a}}$ \}.  
\begin{prop}\label{min_benders}
Benders cut \eqref{bendersMaster:benders} associated with $\bar a ^*$ dominates all \eqref{bendersMaster:benders} associated with $\bar a\in A\setminus \{\bar a^*\}$.
\end{prop} 
For all  $\bar a\in A$, sub-problem can be solved using Ford-Fulkerson algorithm to compute a ``Min-cut'' in the sub-graph of $G$, where each link $a\in A$ with $\sum \limits _ {i\in K} x^i_a = 0$ is removed. \ym{Benders cuts \eqref{bendersMaster:benders} are exponential in number. They are generated dynamically. Indeed, we start solving the master problem without any benders cut. Once an integer solution is found, we solve the sub-problem for every $\bar a\in A$ and generate the associated benders cut. Thanks to Proposition \ref{min_benders}, we identify the dominated cut and if it is violated, we add it to the master problem. We iterate until there is no violated benders cut.   }   

\section{Relaxation of the problem}\label{sec:Problem_relaxation}
In this section, we propose to relax the APCP to obtain a problem that can be solved in polynomial time. First, let us introduce the relaxed problem. 
The relaxation of APCP (RAPCP) consists in replacing the objective $1)$ with the two following objectives:
\begin{itemize}[leftmargin=5.6em]
    \item[\textbf{Objective a}:] maximize the number of disjoint paths, 
    \item[\textbf{Objective b}:] {maximize the minimum capacity over all paths, i.e., $\max \min \limits _{p\in P, a\in p} cap(a)$.}
\end{itemize}
{Note that, objective 2) is still considered in RAPCP.} Hence, RAPCP is equivalent to finding $k'\leq k$ paths totally link-disjoint where objectives are considered in lexicographic order, objective a) $>>$ objective b). 
The approach proposed in this section does not need any bigM value to ensure the lexicographic order of objectives. 

Let us introduce the following notation:
\begin{itemize}
    \item $obj_1(P)$ the value of the objective 1) on the paths $P$. 
    \item $obj_{a}(P)$ the number of disjoint paths from $P$.
    \item $obj_{a,b}(P)$ the value of the objective a) and b) on the paths $P$, i.e.  $|obj_{a}(P)|\times\min_{p\in P} mincap(p)$. 
\end{itemize}

{In the following, we show that the optimal solution of the RAPCP gives a lower bound for the APCP. Thanks to bigM on objective 1), objective 2) can be ignored in this proof. 
\begin{prop}
  $obj_{a,b}(x^{RAPCP}) \leq obj_1(x^{APCP})$
\end{prop}
\begin{proof} 
Let $P_{RAPCP}$ be a solution of RAPCP. Since RAPCP computes totally disjoint paths, we may have $|P_{RAPCP}|<k$. Let $P_{RAPCP}^*$ be the solution obtained from $P_{RAPCP}$ by duplicating the shortest path in $P_{RAPCP}$, $k-k'$ times. It is easy to see that $P_{RAPCP}^*$ is a feasible solution for APCP.\\
Moreover, the quantity of traffic that can be sent over paths of $P_{RAPCP}^*$ is greater than or equal to $obj_1(P_{RAPCP})=\sum_{p\in P_{RAPCP}} mincap(p)$. This latter is also greater than or equal to $obj_{a,b}(P_{RAPCP})=|P_{RAPCP}|\times\min_{p\in P_{RAPCP}} mincap(p)$, as paths of $P_{RAPCP}$ are also link-disjoint. This is enough the show the result. 
\end{proof}
} 

APCP solutions for $k=2$ and $k=3$, shown in Fig.~\ref{APCP:fig} are the same for RAPCP. 
\begin{figure}[!t]
\centering
\includegraphics[page=2,scale=.20]{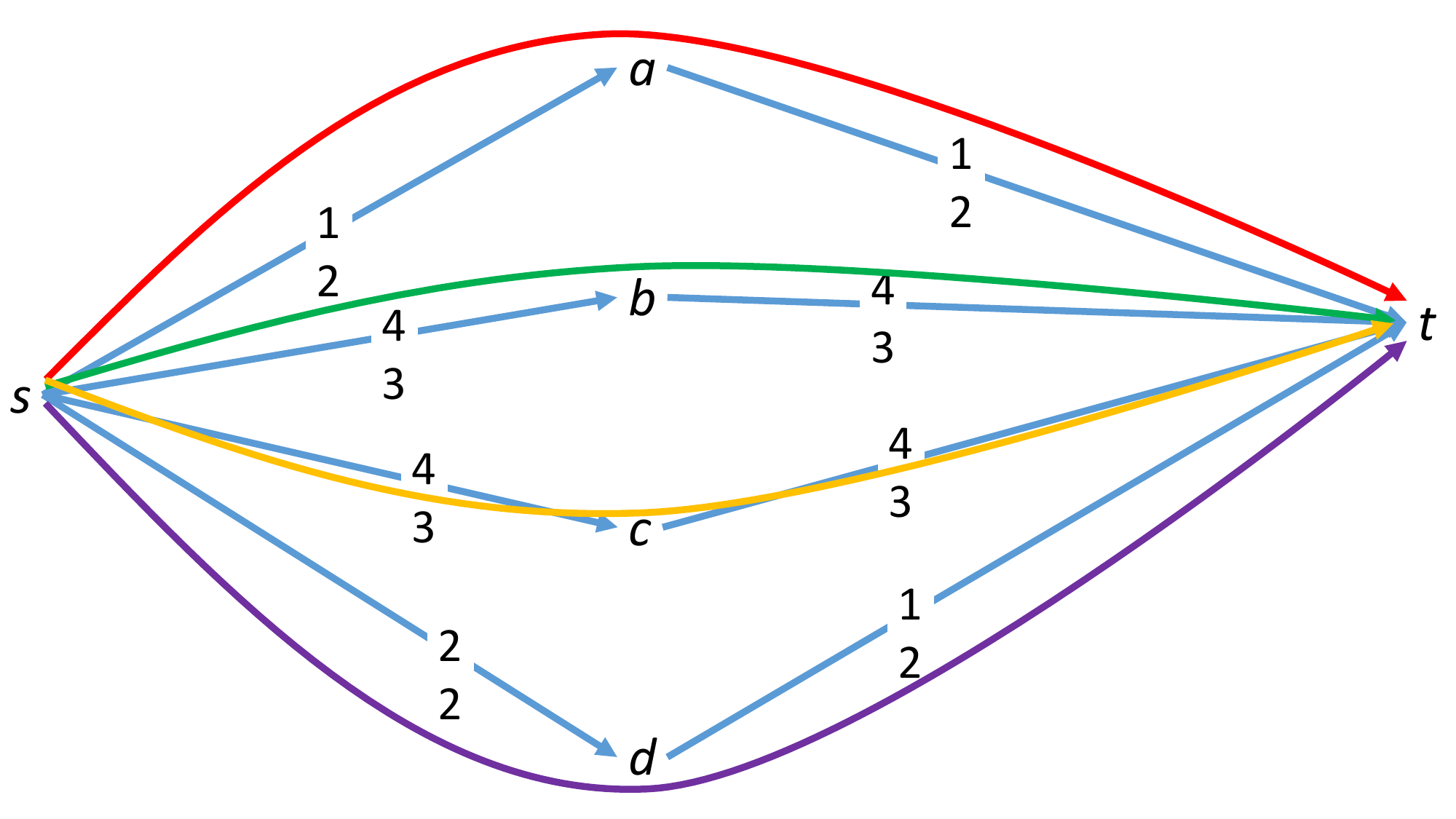}
\caption{Example of solution for APCP (green, yellow, and red paths) and RAPCP (red, purple, and green paths) where $k=3$.}
\vspace{-0.5cm}
\label{RAPCP:fig2}
\end{figure}
In Fig.~\ref{RAPCP:fig2}, we can see different solutions for APCP and RAPCP. The best solution for the APCP for $k=3$ is given by the green, yellow, and red paths where objective 1) is equal to 5 and objective 2) is equal to 18. For the RAPCP, the best solution is given by the red, purple, and green paths where objective a) is equal to 3, objective b) is equal to 2, and objective 2) is equal to 13. For this solution, objective 1) is equal to 4, which is worse compared to the solution found by the APCP.

To solve the RAPCP problem, we propose an exact polynomial algorithm based on the min-cost flow problem (MCFP). The goal is to solve a polynomial number of MCFP to reach the optimal value.


To introduce our algorithm, we first consider only objective a) and objective 2). This algorithm will be a sub-routine of the next algorithm, able to optimally solve the RAPCP.  
We add a dummy link from $s$ to $t$ with a capacity $k$ and a cost equal to $\sum_{a\in A}c_a$.
Let us consider a capacity of 1 on each original link. Solving a min-cost flow problem on this graph, where the flow is set to $k$ from $s$ to $t$, allows finding $k'$ disjoint paths, where $k'=k-f_d$ where $f_d$ is the flow on the dummy link. 
We can solve the RAPCP problem with objective a) and objective 2) with one call of the min-cost flow problem. We call this algorithm RAPCPa2. Let us denote $RAPCPa2(G,c,k)$ a function that returns the disjoint paths found by the  RAPCPa2 algorithm. \sm{As the paths are totally linked disjoint, it is easy to recover the paths from the min-cost flow solution.} 
    


As our objectives are in lexicographic order, by solving $RAPCPa2(G,c,k)$ we get the best solution for objective a). 
Thus, we can focus on objective b) by setting the arc capacity to $k-obj_a(RAPCPa2(G,c,k))$. To manage objective b), Algorithm~\ref{alg:capAll} filters the link with the smallest capacity at each iteration, in order to find the best solution for objectives a), b) and 2).
{\small
\begin{algorithm}
\caption{RAPCP algorithm }\label{alg:capAll}
\KwData{Graph $G=(V,E)$, cost vector $c$, number of paths $k'\leq k$ }
$P^*\gets RAPCPa2(G,c,k)$\;
$\ell_a \gets k-obj_a(P^*)$\;
$BestP \gets P^*$\;
remove all links with a capacity less than or equal to  $obj_{a,b}(P^*)$\;
\While{$P^*\neq \emptyset$}{
  $P^*\gets RAPCPa2(G,c,k)$\;
  \eIf{$k-obj_a(P^*)<\ell_a$}{
  $P^*\gets \emptyset$
  }{
$BestP \gets P^*$\;
  remove all links with a capacity less than or equal to  $obj_{a,b}(P^*)$\;
  }
}
\KwResult{$BestP$}
\end{algorithm}
}
\sm{Remark that RAPCPa2 and RAPCP algorithms can return less than $k$ paths. Thus, to get exactly $k$ paths, we complete with replicas of the shortest path to guarantee having the best objective 2.}


\section{Experimental results} \label{sec:Experimentation}

All the algorithms presented in the previous sections have been implemented in C++, using CPLEX as an LP-solver. They were tested on an Intel(R) Xeon(R) CPU E5-4627 v2 of 3.30GHz with 504GB RAM, running under Linux 64 bits. Only 1 thread has been used. In the following, we present evaluation results on randomly generated instances (topology and demands) \ym{starting from a random tree spanning tree for connectivity and randomly adding new links}. We vary the number of nodes, density, and the number of alternative paths.
We compare four methods: APCP (resp. APCP\_BENDERS) which corresponds to solving model (P) (resp. Benders Decomposition algorithm) using CPLEX, RAPCP (resp. RAPCPA2) to solve RAPCP with objective a) b) and 2) (resp. a) and 2) ) using Lemon\cite{dezsHo2011lemon} to solve the min-cost flow problem.


\sm{We do not consider some basic algorithms like KSP used in \cite{10302752}. In contrast with KSP, the main purpose of our approach is a kind of disjointness maximizing the flow after any failure. }

Due to the computational times of the different algorithms, we considered the following configurations:
\begin{enumerate}
    \item Configuration 1:
    \begin{itemize}
        \item algorithms: all,
        \item number of nodes: $20$,
        \item topology density: $40\%$ and $60\%$,
        \item $k$: $3$ and $6$. 
    \end{itemize}
    \item  Configuration 2:
    \begin{itemize}
        \item algorithms: APCP\_BENDERS, RAPCP and RAPCP2a
    \item number of nodes: $40$,
    \item topology density: $10\%$,
    \item $k$: $3$ and $6$.
    \end{itemize}
\end{enumerate}

For a given congested link and a given destination, we compare the performance of all the algorithms within a time limit of 200 seconds. 

\sm{For each graph, each value of $k$ and each number of alternative paths, we define an instance as a potentially-congested link-destination pair, i.e., we consider all the potential congestions that can introduce rerouting of traffic towards each potential destination. Thus, for each pair graph-number of alternative paths, we get $|A||V|$ instances.}

 \sm{We compare the performance for other Key Performance Indicators (KPI). The first one is the \emph{Cost}, defined as the sum of the costs of the alternative paths. The second one is \emph{Min Surviving Paths}, i.e., over all the possible link failures the minimum number of paths that do not cross a given failure. This represents the worst case of surviving alternative paths after each failure.  The third one is the \emph{Min Max Flow} where we consider, for each failure, the minimum of the maximum flow that can be sent on the remaining graph. The last KPI is the  \emph{Paths disjointness} which represents the number of disjoint paths before any failure. }




{
\begin{figure}[H]
\includegraphics[trim=0cm 0cm 0.7cm 0cm, clip=true, scale=.28]{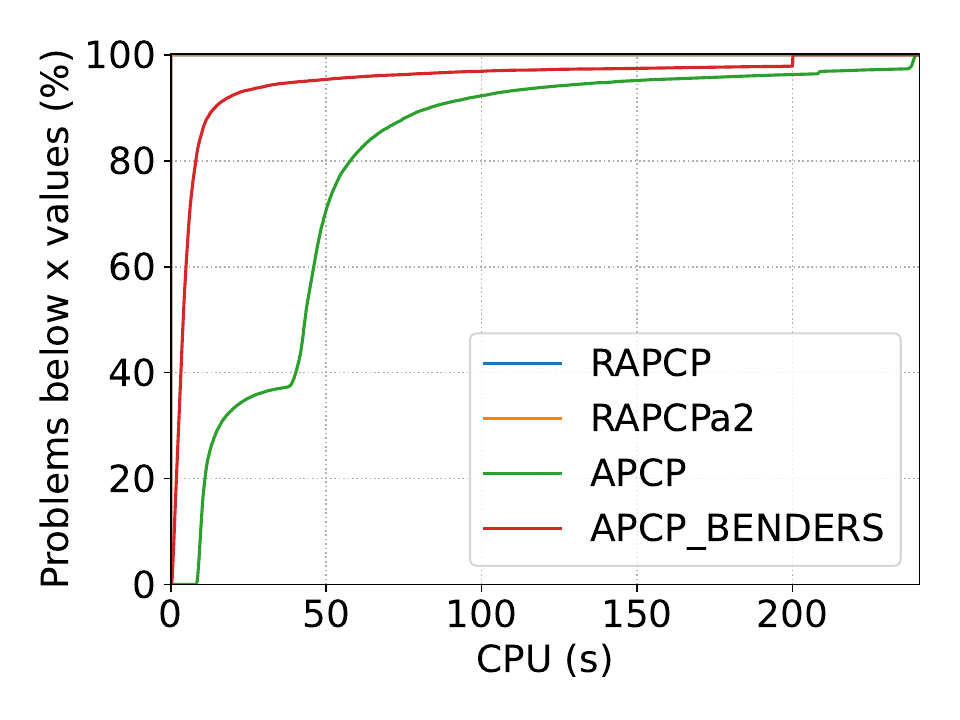}~
\includegraphics[trim=0.7cm 0cm 0cm 0cm, clip=true, scale=.28]{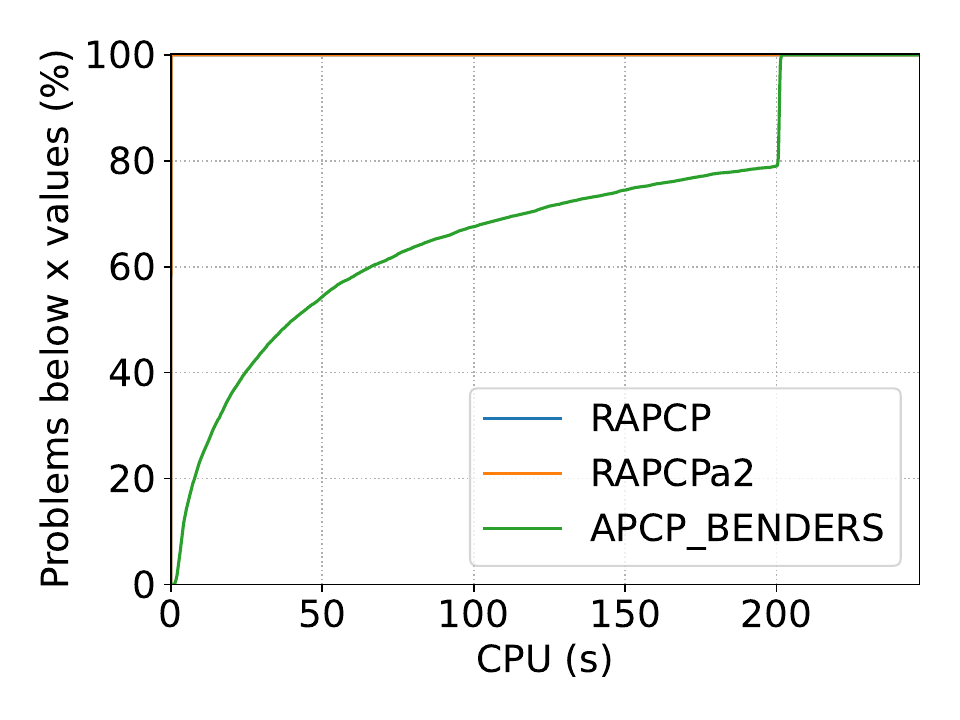}
\caption{Performance Profile comparing the CPU time on Configuration 1 (Left) and Configuration 2 (Right) instances. RAPCP and RAPCP2a are not visible since they reach 100\% in less than a second. }
\label{ALL:cpu20_40}
\end{figure} 
Figure \ref{ALL:cpu20_40} compares the CPU time over all instances of configurations 1 and 2. Clearly, we notice that the APCP method is the slowest one since it requires almost 100 seconds to solve $50\%$ of the instances. The APCP\_BENDERS is much faster since it solves more than $90\%$ of the instances in less than 100 seconds. Clearly, RAPCP and RAPCP2a are faster, solving all instances in less than one second. By comparing the RAPCP and Benders algorithm, the RAPCP is faster, than Benders on average, by 11 seconds on configuration 1 and by 77 seconds on configuration 2.}
\begin{figure}[H]
\includegraphics[trim=0cm 0cm 0.7cm 0cm, clip=true, scale=.28]{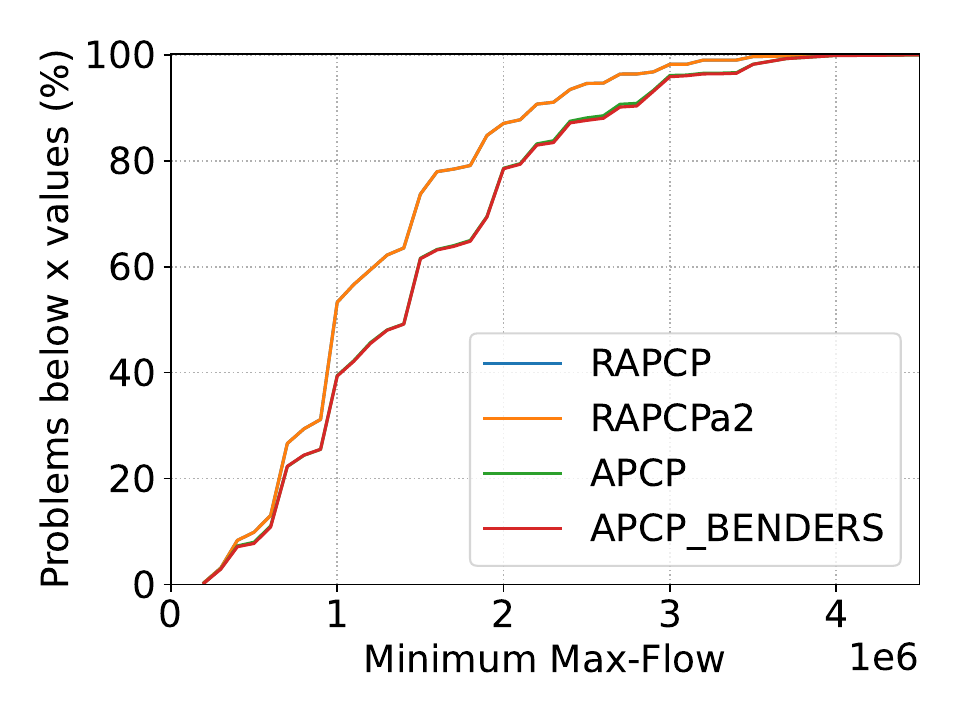}
\includegraphics[trim=0.7cm 0cm 0cm 0cm, clip=true, scale=.28]{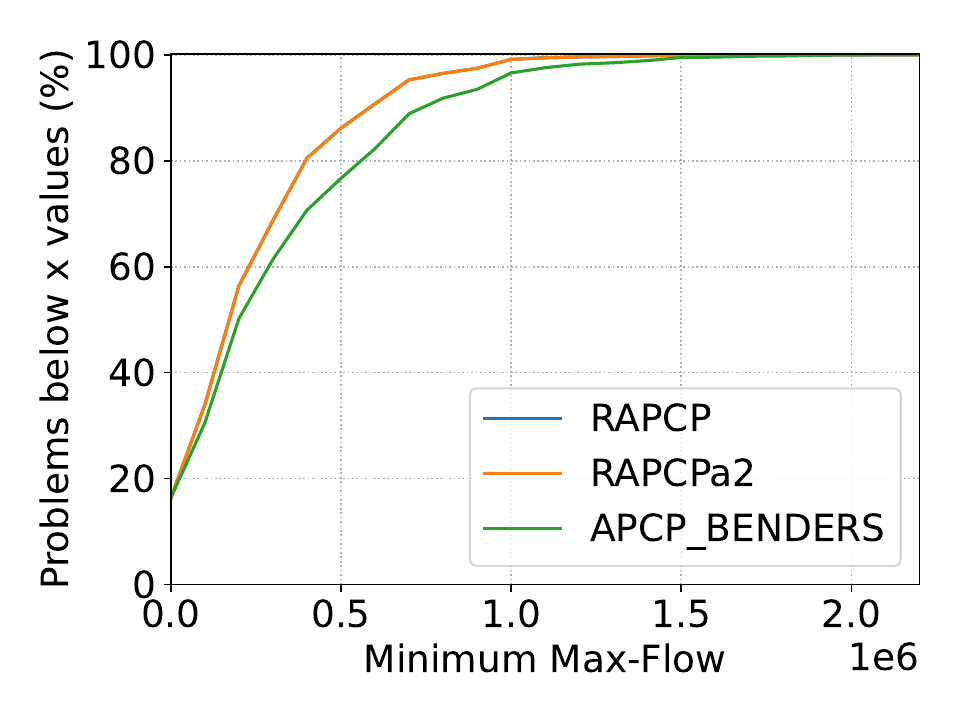}
\caption{Performance Profile comparing the max-flow after 1 link failure on Configuration 1 (Left) and Configuration 2 (Right) instances. RAPCP is not visible, overlapped by RAPCPa2 curves. Similarly for APCP in the left plot.} \label{ALL:maxflow20_40}
\end{figure} 
{Figure \ref{ALL:maxflow20_40} compares the maximum flow after 1 link failure over all instances of configurations 1 and 2. The orange line is always above the red or green line. This means, that the compact and Benders methods provide paths allowing flooding more traffic after any 1-link failure. On average, Benders' paths can flood around 248Gbps more traffic than RAPCP on configuration 1 and around 70Gbps more traffic than RAPCP on configuration 2.}

\begin{small}
\begin{table}[H]\centering
\begin{tabular}{l|c|c}
	& Configuration 1	&Configuration 2 \\
\hline
Cost	&-1.65 \%	&3.21 \% \\
Min Surviving Paths 	&15.04 \%&	23.87 \%\\
Min Max-Flow	&-19.53 \%	&-24.69 \% \\
Paths disjointness	&14.20 \%	&8.20 \% \\
\end{tabular}
\caption{Table displaying the gap between RAPCP and APCP\_BENDERS on Configurations 1 and 2  over several KPIs.}
\end{table}\label{table_kpi}
\end{small}
Table \ref{table_kpi} shows, for each KPI, the gap between two average values: RAPCP and APCP\_BENDERS, i.e., $\frac{\text{Average RAPCP}-\text{Average APCP$\_$BENDERS}}{\text{Average RAPCP}}\times 100$. We notice that the two methods give almost close paths, in terms of objective cost, on both configurations. However, we see a clear advantage of RAPCP on the number of alive paths and the number of disjoint paths after failure. In contrast, APCP allows flooding a much higher flow after link failure, compared to RAPCP.

\section{Conclusion} \label{sec:Conclusion}
{In this paper, we have investigated two variants of the alternative paths computation problem for a distributed congestion mitigation mechanism based on segment routing. For the first variant, i.e., the maximization of the minimum max-flow for any link failure, we proposed a compact formulation to model the problem, and showed that the problem is NP-hard. Then, for the sake of computational time, we developed a Benders decomposition algorithm. \ym{To provide a practical and scalable solution, we proposed a second variant that maximizes the number of disjoint paths and the minimum link capacity, after any link failure. W}e proposed two polynomial algorithms to solve it. We show that the second variant is a relaxation of the first one. Through numerical results, we showed that the first variant allows giving better alternative paths in terms of the quantity of flow that can be transferred after any link failure. However, the second variant gives better paths in terms of the number of disjoint paths and the number of remaining paths after any link failure. The second variant can be a good solution, especially for computation time, which is much faster. The second approach provides a good tradeoff for scalability whereas the first one allows solving optimally small networks.  }
\bibliographystyle{IEEEtran}
\bibliography{papers}

\end{document}